\documentclass{new_tlp}
\pdfoutput=1
\usepackage{epsfig}
\usepackage{makeidx}  
\usepackage{amsmath}
\usepackage{draftwatermark}
\SetWatermarkText{DRAFT}
\SetWatermarkScale{1}

\newtheorem{theorem}{Theorem}

\newcounter{example}
\newenvironment{example}[1][]{\refstepcounter{example}\par\medskip
   \noindent \textbf{Example~\theexample. #1} \rmfamily}{\medskip}

\title[Theory and Practice of Logic Programming]
      {Top-down and Bottom-up Evaluation \\ Procedurally Integrated}
\author[D. S. Warren]
       {DAVID S. WARREN\\
         Stony Brook University, Stony Brook, NY, 11794-4400 USA\\
         \email{warren@cs.stonybrook.edu}\\
         XSB Inc., Setauket, NY, 11733 USA\\
         \email{warren@xsb.com} }
\begin{document}
\maketitle              
\begin{abstract}
This paper describes how XSB combines top-down and bottom-up
computation through the mechanisms of variant tabling and subsumptive
tabling with abstraction, respectively.

It is well known that top-down evaluation of logical rules in Prolog
has a procedural interpretation as recursive procedure invocation
\cite{Kowalski:1986:LP:13060}.  Tabling adds the intuition of
short-circuiting redundant computations \cite{memoing:warren}.  This
paper shows how to introduce into tabled logic program evaluation a
bottom-up component, whose procedural intuition is the initialization
of a data structure, in which a relation is initially computed and
filled, on first demand, and then used throughout the remainder of a
larger computation for efficient lookup.  This allows many Prolog
programs to be expressed fully declaratively, programs which formerly
required procedural features, such as assert, to be made efficient.

This paper is under consideration for acceptance in {\em Theory and
  Practice of Logic Programming (TPLP)}.
\end{abstract}

\section{Introduction}
Top-down and bottom-up algorithms have long been seen as alternative,
even competitive, ways to evaluate Datalog (and Prolog) programs
\cite{Lloyd:1993:FLP:529834,bry-reconciled}, with proponents of each
algorithm arguing for its particular advantages
\cite{Ullman:1989:BBT:73721.73736,Toman95top-downbeats,Tekle:2011:MED:1989323.1989393}.
Systems implement either top-down evaluation, e.g., XSB (by default)
\cite{xsbprolog:tplp} and other Prolog systems that include tabling,
such as YAP \cite{DBLP:journals/tplp/CostaRD12}, Ciao Prolog
\cite{DBLP:journals/tplp/HermenegildoBCLMMP12}, BProlog
\cite{DBLP:journals/tplp/Zhou12} and as a library in SWI Prolog
\cite{SWI-tabling}, or bottom-up evaluation, e.g., LDL++
\cite{LDLplusplus}, Coral \cite{DBLP:journals/vldb/RamakrishnanSSS94}
and LogicBlox \cite{Aref:2015:DIL:2723372.2742796}, among many others.
In this paper we consider the properties and advantages of each
strategy and propose that they be integrated in a single
system, with each having a distinct procedural interpretation.  Then
each can be used when its evaluation strategy is most appropriate and
efficient.  This paper describes how they are integrated in the XSB
tabled logic programming system.

\section{Brief Review of Top-Down and Bottom-up Evaluation}
\label{tdbu-eval}

To briefly review tabled top-down and bottom-up evaluation of Datalog
rule sets, we consider the following well-known example of transitive
closure for a simple graph:

\begin{verbatim}
:- table p/2.
                                   e(a,b).     e(b,c).
p(X,Y) :- e(X,Y).                  e(e,a).     e(c,b).
p(X,Y) :- p(X,Z),e(Z,Y).           e(d,e).
\end{verbatim}
\normalsize

\begin{example}
For query {\tt p(a,A)}, tabled top-down evaluation proceeds as
follows:
\begin{verbatim}
1 p(a,A)            add query p(a,?) to table
2   e(a,A).         resolve 1 with 1st rule
3     () A=b        resolve 2 with fact e(a,b), add answer p(a,b) to table
4   p(a,Z),e(Z,A)   resolve 1 with 2nd rule
5     e(b,A)        resolve 4 with answer p(a,b) from table
6       () A=c      resolve 5 with fact e(b,c), add answer p(a,c) to table
7     e(c,A)        resolve 4 with answer p(a,c) from table
8       () A=b      resolve 7 with fact e(c,b), p(a,b) in table, don't add
\end{verbatim}
\normalsize
\end{example}
  
\begin{example} \label{bu-p2}
Bottom-up evaluation for the same query proceeds in iterations. At
each iteration it uses the program rules and the facts currently known
to infer new facts.

\begin{verbatim}
Iteration 0: infer the program facts:
        e(a,b), e(e,a), e(d,e), e(b,c), e(c,b)
Iteration 1: Use those facts and program rule 1 to infer:
  new:  p(a,b), p(e,a), p(d,e), p(b,c), p(c,b)
Iteration 2: Use rule 2 and previous inferred facts to infer:
  new:  p(a,c) [from p(a,b) and e(b,c)]
        p(e,b) [from p(e,a) and e(a,b)]
        p(d,a) [from p(d,e) and e(e,a)]
        p(b,b) [from p(b,c) and e(c,b)]
        p(c,c) [from p(c,b) and e(b,c)]
Iteration 3: Again use rule 2 with previously inferred facts to infer:
  new:  p(e,c) [from p(e,b) and e(b,c)
        p(d,b) [from p(d,a) and e(a,b)
          (We get duplicates, e.g., p(a,b) from p(a,c) and e(c,b),
           but they aren't added.)
Iteration 4: Again use rule 2
        p(d,c) [from p(d,b) and e(b,c)]
Iteration 5: Nothing new to infer, so stop
\end{verbatim}
\normalsize
\end{example}

These logs show the difference: Top-down uses the binding on the first
argument of the query to guide its search, so here it never
has to look at nodes {\tt d} and {\tt e} and their successors.
Bottom-up computes the entire transitive closure relation, and would
then select from it those p-tuples with value {\tt a}
in the first position.  This might be more efficient if we were to ask
many queries to the {\tt p/2} relation, with different first
arguments, since the {\tt p} relation would be computed only once, and
then every subsequent query would require only a simple lookup in that
table.

In this particular case, there is not much to gain from computing the
entire relation first, because the top-down strategy (saving tables
across queries) would efficiently compute only the new tuples required
for the current query, using the previously computed answers as
needed.  But there are (many) cases in which that efficient
incremental computation is not what happens with multiple top-down
queries to the same relation.

\section{Example of Need for Bottom-Up Evaluation}

\begin{example} \label{corpus}
Consider a simple system with a corpus of sentences,
which allows a user to input a sentence and find all the sentences in
the corpus that share a word with the input sentence.  The system uses
a relation of strings, called {\tt corpus(CorSent)}, which contains
the sentences of the corpus, and defines a predicate {\tt
  share(+InSent,-CorSent,-Word)}, that, given an input sentence,
returns each corpus sentence that shares a word with it, and the shared
word.  It can be defined using a helper tokenizing
predicate, {\tt scan/2}, as follows:

\begin{verbatim}
share(InSent,CorSent,Word) :-
    scan(InSent,InWordList),
    corpus(CorSent),
    scan(CorSent,CorWordList),
    member(Word,InWordList),
    member(Word,CorWordList).
\end{verbatim}
\normalsize

Top-down evaluation by Prolog proceeds as follows:
Given a string as input, {\tt scan/2} tokenizes it,
providing the list of words it contains.  Then for each sentence in
the corpus, we tokenize it, and return the corpus sentence for
each word in the input sentence token list that is also in the corpus
sentence token list.
\end{example}

This is a reasonable approach for a single query, but if we want to
answer multiple similar queries to the same corpus, this will
re-tokenize every string in the corpus for each query, which would be
extremely redundant and inefficient: Answering a
single query for one input sentence essentially requires almost all
the work (i.e., tokenizing every corpus sentence) that is necessary to
answer a query for any input sentence.  In the presence of multiple
queries, it would be much better to construct a table of
the corpus sentences and their tokens only once, and then use this
table to answer each input query.  We can simply reorder and fold the
above program clause to create and use that {\tt corpus\_word/2}
relation as follows:

\begin{example}
\begin{verbatim}
share(InSent,CorSent,Word) :-
    scan(InSent,InWordList),
    member(Word,InWordList),
    corpus_word(CorSent,Word).
    
corpus_word(CorSent,Word) :-
    corpus(CorSent),
    scan(CorSent,CorWordList),
    member(Word,CorWordList).
\end{verbatim}
\normalsize
\end{example}

\noindent This separates out the relation we want to generate and reuse, but
notice that variant tabling does not solve our redundant computation
problem.  The predicate {\tt corpus\_word/2} will be called with its
second argument, {\tt Word}, bound.  So for each different word in the
input sentence, {\tt corpus\_word/2} will be entered, thus re-scanning
every corpus sentence for each new input word.  This is even worse than
the previous unfolded program.  

However, if we can ensure that the {\tt corpus\_word/2} definition is
evaluated bottom-up, generating a complete table that is indexed on
its second argument and that table is used for every call from {\tt
  share/3}, then the redundant computation will be eliminated.

So this is an example of a situation in which we want to build (i.e.,
initialize) a table containing more answers than we immediately need:
we want answers computed for all corpus words, even though we
initially need only those for a single word.

In procedural terms, we understand the bottom-up evaluation of {\tt
  corpus\_word/2} as building an intermediate data structure to more
efficiently evaluate various subqueries that will be needed later.
This is a common requirement in algorithm development.  Experienced
Prolog programmers normally implement such requirements
non-declaratively by initially computing the necessary tuples and
adding them to a dynamic predicate using the Prolog built-in {\tt
  assert/1}.  But with bottom-up evaluation, such programs can be
purely declarative and yet have the equivalent computational
properties.

\section{Specifying Bottom-Up Evaluation in XSB}

The question now is how to indicate to the evaluation system that the
first call to {\tt corpus\_word/2} is to cause it to be evaluated
bottom-up to fill the table that is to be indexed on its second
argument, and how the evaluation system will efficiently carry out
that evaluation.  

XSB uses Subsumptive Tabling With Abstraction (STWA) to carry out the
bottom-up evaluation.  Normally tables in XSB are variant tables,
i.e. an entry in a variant table is used to satisfy a subsequent goal
if that goal is a variant of the goal that generated the table entry,
i.e., equal up to change of variable.  In subsumptive tabling
\cite{Johnson1999,Tries@JLP,ejohnson-tapd,EJohnson-thesis,DBLP:journals/corr/abs-1107-5556},
a table entry is used to satisfy a subsequent goal if the subsequent
goal is subsumed by the generating goal, i.e., is an instance of the
generating goal.  Note that in such a case the answers to the subsumed
goal will (eventually) be in the table of answers for the previously
encountered subsuming goal.\footnote{Not all predicates defined in
  Prolog can be correctly tabled.  For example, variant tabled
  predicates should not use built-ins that change the global state,
  such as {\tt assert/1} or {\tt retract/1}; subsumptively tabled
  predicates, in addition, should not use built-ins that depend on the
  instantiation state of variables, such as {\tt var/1} and many uses
  of the cut (!) operation.  } The programmer may specify, using a
{\tt table\_index/2} directive, that when a subsumptively tabled goal
is called, the calling goal should be abstracted before it is actually
invoked, so that a table for the more general call is constructed, and
then only the answers unifying with the original unabstracted goal are
returned to the initial call.  In this case, since a general table is
constructed, a later subgoal that is subsumed by that more general
goal will use this general table to satisfy the later subgoal and not
make the specific subgoal call to do clause resolution.

The {\tt table\_index/2} directive was recently added to XSB to
support the creation of subsumptive tables with multiple arbitrary
indexes \cite{xsbmanual}.  An index\footnote{These index specifications are the
  same as those used for dynamic predicates in XSB.} on a
single argument of the table is specified by just giving an integer
that indicates the position in the tabled predicate of the argument to
be indexed.  Joint indexes, i.e. indexes on multiple arguments, are
specified by terms made by separating the multiple position integers
by the symbol $+$.  Multiple different indexes are provided in a list.
For example, say we want four indexes on a subsumptive table for {\tt
  p/4}, so that a call to {\tt p/4} (after its table is built) first
checks if the call has arguments 1 and 2 bound and if so uses a joint
index on those arguments; and if not, then checks if it has just
argument 1 bound and if so uses an index on that; and if not then
checks if arguments 2, 3 and 4 are bound and if so uses a joint index
on those positions; and if not then checks if argument 4 is bound and
if so uses an index on that position; and if not, then throws an error
(because that would be an unindexed call and we should be notified).
So this would require the directive:
\begin{verbatim}
:- table_index(p/4,[1+2,1,2+3+4,4]).
\end{verbatim}
The {\tt table\_index/2} directive not only specifies the indexes on
the resulting subsumptive table but also what abstraction should be
done to the initial call to the subsumptively tabled predicate.  The
assumption is that every index will have some calls that use it.
Since we want to compute the tuples for only one abstracted call to
the tabled predicate, every position that might be unbound on some
particular call must be abstracted in the initial subsumptive call, so
that the one table will contain all answers for a call not bound on
that position.  This implies that the arguments that should be bound
in the abstracted call to the subsumptively tabled predicate are those
that appear in every specified index; all others must be variables.
For the {\tt p/4} example, there are no such arguments, so the first
call will be fully abstracted, i.e., all arguments will be
variables.\footnote{In this case we call it ``subsumptive tabling with
  {\em full} abstraction''.}

There are situations where we want the initial call to a subsumptively
tabled predicate to be abstracted to the fully open call, even if
there are positions appearing in all index specifications.  For
example, we may need only one index but want full abstraction, as in
our example for {\tt corpus\_word/2}.  In such a case, we must use $0$
as an index specifier, meaning ``no index'', and it is placed last.

Thus for the {\tt corpus\_word/2} subsumptive table, we want:

\begin{verbatim}
:- table_index(corpus_word/2,[2,0]).
\end{verbatim}

\noindent which declares that the predicate {\tt corpus\_word/2} should be
computed bottom-up.  We want an index on the second argument and the
first call should be abstracted to the call with all variables.  Since
every {\tt corpus\_word/2} subgoal is subsumed by this most general
goal, it will never be called after the first time, but the table that
is constructed at the first call will be used to satisfy all those
subsequent calls.  The first call to {\tt corpus\_word/2} will build
(or initialize) the table (or data structure), and from then on every
call will directly access that table; and in our case use the index
for constant time lookup.

The directive {\tt table\_index/2} is implemented in XSB by a program
transformation, which we briefly describe by example.  The following
code (slightly modified for clarity) is generated by the compiler for
the {\tt p/4} example above, \\
\verb|:- table_index(p/4,[1+2,1,2+3+4,4])|:


\begin{verbatim}
:- table p1234/4, p4231/4 as subsumptive.
p(A,B,C,D) :-
    nonvar(A) -> p1234(A,B,C,D)
    ; nonvar(D) -> p4231(D,B,C,A)
    ; table_error('Illegal Mode in call to p/4').
p1234(A,B,C,D) :-
    var(A),var(B),var(C),var(D) -> p_base(A,B,C,D)
    ; p1234(E,F,G,H),E = A,F = B,G = C,H = D.
p4231(A,B,C,D) :-
    var(D),var(B),var(C),var(A) -> p1234(D,B,C,A)
    ; p4231(E,F,G,H),E = A,F = B,G = C,H = D.
p_base(A,B,C,D) :- ... original p rules ...
\end{verbatim}
\normalsize

\noindent {\tt p1234/4} and {\tt p4231/4} are versions of {\tt p/4}
with its arguments permuted.  Subsumptive tables implement trie
indexing, thus providing indexing on any initial sequence of
arguments.  Here two tables are needed: {\tt p1234/4} for indexes $1$
and $1+2$, and {\tt p4231/4} for indexes $4$ and $2+3+4$.\footnote{The
  complier determines the minimum number of tables required to cover
  all the declared indexes.}  The first clause tests the calling mode
to select which tabled permutation to call.  Each permutation
predicate will normally be entered twice, once with some arguments
bound, which will result in a (second) call to itself with all
arguments unbound.  From then on every other call will be subsumed by
that most general call and so its answers will be served from that
generated table.  Notice that {\tt p4231/4} calls {\tt p1234/4} and
not {\tt p\_base/4} and so will be filled by getting its answers from
the table of {\tt p1234/4}.  Thus the original clauses that define
{\tt p/4} (now in {\tt p\_base/4}) will be called only once.

\section{Intermixing Top-down and Bottom-up evaluation}

It is possible to combine top-down filtering with bottom-up
evaluation.

\begin{example}
Say our corpus of sentences in Example \ref{corpus} is
partitioned to indicate those from a single book.  We would have
another argument to {\tt corpus/1} now: \\
{\tt corpus(BookISBN,Corent)}.  And top-level predicate: \\
{\tt share(InSent,BookISBN,CorSent)}, and sub-predicate: \\ {\tt
  corpus\_word(BookISBN,CorSent,Word)}, defined in terms of {\tt
  corpus/2}.

\begin{verbatim}
share(InSent,BookISBN,CorSent) :-
    scan(InSent,InWordList),
    member(Word,InWordList),
    corpus_word(BookISBN,CorSent,Word).
    
corpus_word(BookISBN,CorSent,Word) :-
    corpus(BookISBN,CorSent),
    scan(CorSent,CorWordList),
    member(Word,CorWordList).
\end{verbatim}
\normalsize

We assume that {\tt corpus/2} is indexed on its first field, BookISBN,
and require that BookISBN be bound on a call to {\tt share/3}.  Then we
would invoke partial bottom-up evaluation with the declaration:

\begin{verbatim}
:- table_index(corpus_word/3,[1+3,1]).
\end{verbatim}

\noindent
This will abstract a call to {\tt corpus\_word/3} to {\tt
  corpus\_word(+,-,-)}, where {\tt +} indicates a bound argument and
{\tt -} a variable argument.  Positions that appear in {\em all} indexes
in the {\tt table\_index} declaration, here {\tt 1}, will not be
abstracted; all others will be.  Now the sentences of each book will
be processed bottom-up independently.  When a new book is queried for
the first time, the table for words in sentences for that book will be
computed bottom-up and stored in the table.  Subsequent queries to the
same book will directly use the constructed table.  So this is an
example in which demand-driven computation (here with respect to {\tt BookISBN})
can be combined with data-driven computation (here for the set of
corresponding $<CorSent,Word>$ pairs) to evaluate a query (here to
{\tt share/3}.)
\end{example}

That a data-driven computation can use a demand-driven subquery has
already been seen in the original Example \ref{corpus}.  While we
haven't shown the code for the {\tt scan/2} predicate, it is naturally
defined using demand-driven computation.  E.g., a DCG defining the
scanner would use top-down evaluation of a subquery.  So the original
example shows how the bottom-up evaluation of {\tt corpus\_word(-,-)}
uses the top-down evaluation of {\tt scan(+,-)}.

\section{Procedural Interpretation of Bottom-up Evaluation}

To understand top-down evaluation of rules, one intuitively thinks of
subgoal solving as nondeterministic procedure invocation, with
backtracking to explore alternate computation paths.  To understand
top-down evaluation with tabling, one thinks of procedure invocation
but with short-circuiting for previously encountered procedures (with
the same parameters).  To understand bottom-up evaluation as provided
by Subsumptive Tabling With Abstraction (STWA) through the {\tt
  table\_index/2} directive, one thinks of indexed table construction
(or initialization) from which future retrieval will be done.

A bottom-up computation is appropriate when there is a predicate that
is efficiently computable in one mode and yet it will be needed
multiple times in another mode.  For example, {\tt corpus\_word(+.-)}
of Example \ref{corpus} can be efficiently evaluated in that
direction, i.e., given a sentence find the words in it.  But the
application requires multiple calls of the mode {\tt
  corpus\_word(-,+)}, i.e., given a word find the sentences that
contain it.  In this situation, bottom-up evaluation is called for.

The most basic example of this phenomenon is indexing itself.  For
example, if one wants to find the value associated with a particular
key in an unordered sequence of key-value pairs, one must search the
pairs sequentially, which on average takes time linear in the number
of pairs.  With the same time complexity to search for one item, one
can (with perfect hashing) build a hash table which allows {\em all}
such key lookups to be done in constant time.

\section{STWA for File Input}

Part of the attractiveness of logic programming is that it treats
programs (usually mostly rules) and data (usually facts) the same, but
in most large systems there is an important pragmatic difference
between them: Programs are rules (and facts) that are known statically
at compile time; data are facts (and maybe rules) that are not known
until runtime when they are retrieved from a file.  Facts are normally
accessed through a filename passed to an executing program, where the
records of the file are initially read and asserted into a predicate
in memory, which is used by the program rules.  It is not clear how
this could be made fully declarative.  But consider the following use
of STWA:

\begin{example} \label{file-example}
\begin{verbatim}
:- table_index(emp_data/4,[1+2,1]).

emp_data(FileName,EmpId,Name,Addr) :-
    open(FileName,read,InpStream),
    repeat,
    read(InpStream,Term),
    (Term == end_of_file
     -> !, close(InpStream), fail
     ;  Term = emp(EmpId,Name,Addr)
    ).
\end{verbatim}
\normalsize

\noindent This definition uses low-level procedural Prolog operations
to read the records of a file containing {\tt emp/3} facts and return
them as answers to a call to {\tt emp\_data/4}, where the first
argument of the call is the name of the file to read.  Because of the
{\tt table\_index/2} directive, the first call will be abstracted to
{\tt emp\_data(+,-,-,-)}.  (Recall that arguments in the intersection
of all the indexed arguments, here argument 1, must be bound in any
call, and will not be abstracted for the initial table-filling call.)
So with this definition, a programmer can get access to the tuples in
any employee file (one that contains {\tt emp/3} facts) by invoking a
simple subgoal, with the first argument providing the name of the
file.  There is no need for the programmer to make an explicit call
(e.g., Prolog's {\tt ensure\_loaded/1}) to initialize the {\tt
  emp\_data/4} predicate using Prolog's nonlogical {\tt assert/1}.
The first call will automatically read the file in, building an index
on the first and on the first and second arguments.  Since XSB uses
tries to represent tables \cite{Tries@JLP}, this is here just a single
trie, and the {\tt Filename} is stored only once in the table, so the
overhead of using the 4-ary predicate is essentially zero.

Of course a programmer could read different sets of employee facts
from different files simply by invoking the emp\_data goal with
different file names.
\end{example}

Generalizing this approach, we can make the entire file system a part
of the data-space of a Prolog program.  Definitions such as in Example
\ref{file-example} simply map (portions of) the file system into
Prolog fact-defined predicates.  We can make this easier for the
programmer by providing a single ``file-system'' predicate that views
the entire file system as a single relation.  For example, consider a
system predicate \cite{xsbmanual}: {\tt
  data\_records(+FileName,+FileFormat,?RecordTerm)}, which, given a
filename and a file format, reads the file, parses each line/record to
a Prolog term as specified by {\tt FileFormat}, and returns each such
term nondeterministically in {\tt RecordTerm}.

Now a Prolog programmer can use {\tt table\_index/2} and more easily
define {\tt emp\_data/4} as:

\begin{verbatim}
:- table_index(emp_data/4,[1+2,1]).
emp_data(FileName,EmpId,Name,Addr) :- 
    data_records(FileName,read,emp(EmpId,Name,Addr)).
\end{verbatim}

\noindent to cause the {\tt emp\_data/4} predicate to be loaded from the
named datafile on its first call.  With this framework, the Prolog programmer has
declarative access to all data in the file system, without having to
worry explicitly about how and when to load it into memory.

\section{Does Subsumptive Tabling with Abstraction (STWA) Really Do Bottom-Up Evaluation?}

Subsumptive tabling with abstraction has the properties of bottom-up
evaluation, described above, of generating a full table once and then
using the resulting table for subsequent queries.  But does it
actually build that table using the well-understood bottom-up,
forward-chaining, data-directed algorithm, as described in Section
\ref{tdbu-eval}?  The short answer is yes (mostly), which we now
demonstrate.

\begin{example}
Consider how XSB would compute the transitive closure of Example
\ref{bu-p2} under STWA.  The program is:

\begin{verbatim}
:- table_index(p/2,[1,0]).
                                   e(a,b).     e(b,c).
p(X,Y) :- e(X,Y).                  e(e,a).     e(c,b).
p(X,Y) :- p(X,Z),e(Z,Y).           e(d,e).
\end{verbatim}
\normalsize
Instrumented to produce output showing the engine's progress (see
\ref{progs}), we get (with annotations added for clarity):
\begin{verbatim}
| ?- p(a,A).
(Iteration 1)
enter rule 1: p(X,Y) 
from fact: e(a,b) infer p(a,b)
from fact: e(b,c) infer p(b,c)
from fact: e(e,a) infer p(e,a)
from fact: e(c,b) infer p(c,b)
from fact: e(d,e) infer p(d,e)
(Iteration 2)
enter rule 2 p(X,Y) 
from table: p(a,b) and fact e(b,c) infer p(a,c)
from table: p(b,c) and fact e(c,b) infer p(b,b)
from table: p(e,a) and fact e(a,b) infer p(e,b)
from table: p(c,b) and fact e(b,c) infer p(c,c)
from table: p(d,e) and fact e(e,a) infer p(d,a)
from table: p(a,c) and fact e(c,b) infer p(a,b)
from table: p(b,b) and fact e(b,c) infer p(b,c)
from table: p(e,b) and fact e(b,c) infer p(e,c)
from table: p(c,c) and fact e(c,b) infer p(c,b)
from table: p(d,a) and fact e(a,b) infer p(d,b)
(Iteration 3)
from table: p(e,c) and fact e(c,b) infer p(e,b)
from table: p(d,b) and fact e(b,c) infer p(d,c)
(Iteration 4)
from table: p(d,c) and fact e(c,b) infer p(d,b)
A = b;
A = c;
no
| ?- 
\end{verbatim}
\normalsize
This log closely follows that of Example \ref{bu-p2}: the ``enter
rule'' lines show that each rule (and thus the predicate) is
entered only once.  From then on everything is returned through the
table for {\tt p(X,Y)}.  The only differences are the repeated answers
in Iteration 2, where duplicate answers are computed but not added to
the table.  (Were we more detailed in our description of Example
\ref{bu-p2}, these would have appeared there, too.)
\end{example}

Intuitively, what happens here is that the initial call to {\tt
  p(a,A)} is fully abstracted to {\tt p(X,Y)} and that call is made.
The first rule returns all its answers to the table, which we have
labeled ``Iteration 1''.  Then the second rule is entered and the open
call to {\tt p(X,Z)} (the first subgoal of that second rule) is made,
found to have previously been made, and suspended, to return answers
from the table. The rest of the log, of iterations 2 through 4,
consists of answers from the table being returned to that call, joined
with facts e/2 and the resulting answers added (if new) to the table.
So computation across the bodies of rules may be suspended to wait for
new answers to show up in subsumptive tables.  This, in effect, uses
XSB's scheduling queues to get the equivalent of semi-naive bottom up
computation using pointers.  This computation directly mirrors the
bottom-up evaluation of the query {\tt p(X,Y)}.

\begin{theorem}
Assume H is a Horn program and P a predicate in H such that all
predicates in H are reachable from P in the predicate call graph of
H. (I.e., the goal predicate potentially depends on all predicates.)
Assume that for every rule body in H, there is some instance of it
that is true in the least model of H.  Then evaluation of predicate P
in H under Subsumptive Tabling With {\em Full} Abstraction (STWFA) is
equivalent to a bottom-up computation of the least model of H.
\end{theorem}

\begin{proof}
See \ref{app-proof}
\end{proof}

\section{Bottom-up Proofs for Propositional Horn Clauses in Linear Time}

For a deeper understanding of XSB's STWA and bottom-up evaluation, we
consider the well-known three-line meta-interpreter, applied to
propositional programs.

\begin{example}
Consider the (folded variation of the) classic meta-interpreter of
Prolog programs:

\begin{verbatim}
interp(true).
interp((A,B)) :- interp(A), interp(B).
interp(G) :- interpAtom(G).
    
interpAtom(G) :- (G <- Body), interp(Body).
\end{verbatim}
\normalsize
\end{example}

It is well-known that if we add \verb|<-/2| facts to define a
(propositional) Horn clause program and evaluate {\tt interp(Goal)}
for some {\tt Goal} in Prolog, this will carry out the top-down Prolog
evaluation of {\tt Goal} with respect to the program.  Similarly, if we
add the declaration:

\begin{verbatim}
:- table interpAtom/1. 
\end{verbatim}
\normalsize

\noindent this will carry out top-down evaluation with (variant)
tabling.

Now consider instead adding the declaration:

\begin{verbatim}
:- table_index(interpAtom/1,[1,0]).
\end{verbatim}
\normalsize

\noindent and posing the query: {\tt interpAtom(Prop)}, with {\tt
  Prop} a variable.  Recall that this declaration causes XSB to use
subsumptive tabling.  There is no abstraction necessary for this query
since it is already fully abstract, so all subsequent calls to {\tt
  interpAtom/1} will be served from the table constructed by that
initial call.  The claim is that this will cause XSB to evaluate the
propositional program in {\tt (<-)/2} using a bottom-up algorithm.  To
see this, consider how this query will be processed.  Note that since
there is only the one call to {\tt interpAtom(\_)}, there can be no
propagation of demand through (partially) instantiated arguments.  The
initial goal is {\tt interpAtom(\_)}, the open call.  So this initial
call, being subsumptively tabled, will add the call to the table and
invoke the predicate code to do clause resolution.  Now every other
call to {\tt interpAtom(P)} for any particular proposition {\tt P}
(here only the one in the third clause of {\tt interp/1}) will not
invoke the code for {\tt interpAtom/1}, but will return answers that
have previously been added to the table and will suspend waiting for a
new answer to show up in the table to then be returned.  The first
answers returned cannot require a recursive call to {\tt
  interpAtom/1}, so they can come from only the first clause of {\tt
  interp/1}, and those will be the facts of the program in {\tt
  (<-)/2}.  So the first answers to show up in the table for {\tt
  interpAtom/1} will be the program facts, just as in bottom-up
evaluation.  Then those facts can be returned to the recursive call to
{\tt interpAtom/1} (in the third clause of {\tt interp/1}.)  So rules
all of whose body atoms are facts will have their recursive calls to
{\tt interpAtom/1} satisfied (since they will be in the {\tt
  interpAtom(\_)} table on which they are waiting), and will succeed
returning the propositions in their heads, to be added to the table
for {\tt interpAtom(\_)}.  And so forth.

We can see the order that XSB generates the propositional answers to
our program above by executing the above meta-interpreter (with logging
operations, as shown in \ref{progs}):

\begin{verbatim}
p <- q,v,r,s.          s <- true.
p <- q,s,t.            u <- s,p,v,r.
q <- u,r.              u <- r,t.
q <- q,t,v.            t <- true.
r <- s.
\end{verbatim}
\normalsize

\noindent where we get (lightly edited):

\begin{verbatim}
| ?- interpAtom(p).
1.  Var demanded
2.  p <- q,v,r,s   initial clause
3.  p <- q,s,t     initial clause
4.  q <- u,r       initial clause
5.  q <- q,t,v     initial clause
6.  r <- s         initial clause
7.  s <- true      initial clause
8.  s              from true and s <- true
9.  u <- s,p,v,r   initial clause
10. u <- p,v,r     from s and u <- s,p,v,r
11. u <- r,t       initial clause
12. t <- true      initial clause
13. t              from true and t <- true
14. r              from s and r <- s
15. u <- t         from r and u <- r,t
16. u              from t and u <- t
17. q <- r         from u and q <- u,r
18. q              from r and q <- r
19. q <- t,v       from q and q <- q,t,v
20. q <- v         from t and q <- t,v
21. p <- s,t       from q and p <- q,s,t
22. p <- t         from s and p <- s,t
23. p              from t and p <- t
24. p <- v,r,s     from q and p <- q,v,r,s
25. u <- v,r       from p and u <- p,v,r
\end{verbatim}
\normalsize

Each line (from 2 on) shows a derived clause that is either an initial
program clause, or derived from a previous clause by removing the
first proposition in its body if it is proved (or is the constant
'true').  A proposition is proved if it is the head of a derived
clause with an empty body.  Each step is the result of a unit
resolution, whose participating clauses are indicated on the line.
Unit resolution generates a form of bottom-up propositional reasoning
for Horn clauses.

\vspace{0.15cm}
Performing iterative bottom-up evaluation of this program, we get:
\begin{verbatim}
  Iteration 0:
    s, t: program facts
  Iteration 1: 
    r: from r <- s
  Iteration 2:
    u: from u <- r,t
  Iteration 3:
    q: from q <- u,r
  Iteration 4:
    p: from p <- q,s,t
\end{verbatim}
\normalsize
\noindent And we see these propositions showing up in exactly this
order in the XSB STWA log above.

\section{Performance of XSB STWA for Bottom-Up Evaluation}

In this section we show XSB's performance on bottom-up evaluation of a
propositional program using a meta-interpreter.  The system used is
XSB Version 3.8.0 on a Mac Pro with an i7-4870HQ CPU running at
2.5-GHz with 16 GB of RAM and the 64-bit Windows 10 Operating System.
The program we use is a ``triangular'' program, of the form:

\begin{verbatim}
p1 :- p2,p3,p4,p5.
p2 :- p3,p4,p5.
p3 :- p4,p5.
p4 :- p5.
p5.
\end{verbatim}
\normalsize

This program is an example with 5 rules (and facts) and 15 proposition
occurrences.  The meta-interpreter is:

\begin{verbatim}
interp_goal(true) :- !.
interp_goal((G1,G2)) :- !, interp_atom(G1), interp_atom(G2).
interp_goal(G) :- interp_atom(G).

:- table interp_atom/1.
interp_atom(G) :- interp_atoms(G).

:- table_index(interp_atoms/1,[0]).
interp_atoms(G) :- (G <- Gs), interp_goal(Gs).
\end{verbatim}
\normalsize

Notice that we use {\em both} a variant table (on {\tt
  interp\_atom/1}) {\em and} a subsumptive table (on the equivalent
{\tt interp\_atoms/1}.)  This is to improve efficiency, since lookup
and retrieval from a variant table in XSB is significantly faster than for a
subsumptive table.\footnote{Without the additional variant table, the
  evaluation is about 25 times slower, showing significant room for
  improvement in XSB's implementation of subsumptive table goal
  lookup.  Evaluation of this program with only variant tabling is
  about twice as fast; not surprising since that evaluation is also
  linear, it doesn't duplicate tables, and XSB is optimized for
  variant tables.  } And in these triangular programs there are many
lookups of the same variant call; each proposition is looked up once
per occurrence in the body of a rule, so on average there are
approximately {\tt NumRules/2} lookups per proposition, so lookups
completely dominate the computation time.

Figure \ref{buperftable} shows the cpu time for bottom-up evaluation
of various-sized triangular programs.  $N$ corresponds to the
number ($*10^6$) of proposition occurrences in the program, or more
precisely, the number in the largest triangular program that has
$N*10^6$ or fewer proposition occurrences.

\begin{figure}
\includegraphics[width=12cm,height=8cm, trim=1.5cm 3.6cm 0cm 4.2cm, clip=true]{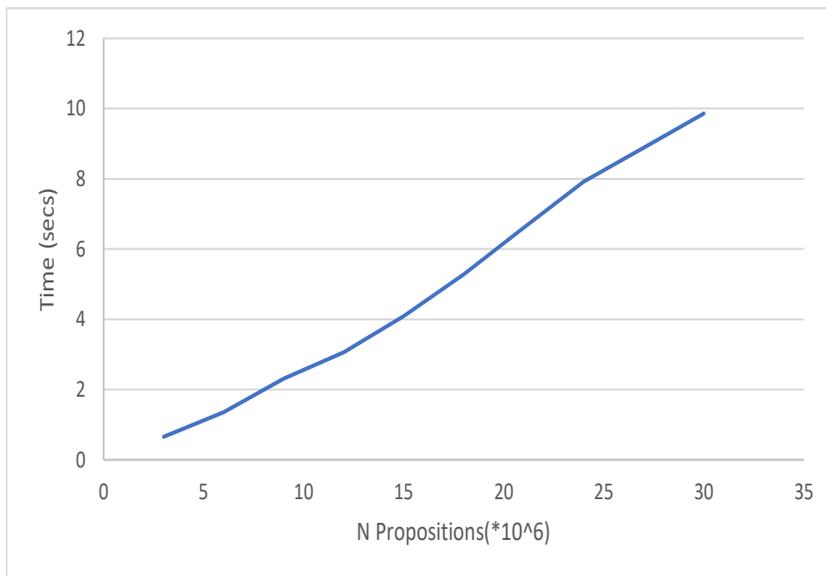}
\caption{Times for Bottom-up Evaluation of Triangular Programs}\label{buperftable}
\end{figure}

Times are just for the meta-interpretation of the triangular programs;
the time for generating the {\tt (<-)/2} facts for the programs is not
included. The graph is essentially a straight line, showing the
linearity of the algorithm.  The propositional Horn clause
satisfaction problem is known to be linear in the number of
proposition occurrences \cite{Dowling1984LinearTimeAF}, and this
meta-interpreter is indeed linear on all propositional logic programs.

Triangular programs have no loops so pure top-down evaluation without
tabling will terminate, but will be very slow due to many redundant
computations.  For example to evaluate without tabling the single
query, {\tt p1}, to a triangular program with 325 proposition
occurrences takes c. 4.1 seconds, and with 465 occurrences c. 126
seconds.  This contrasts with the 15,000,000 occurrences in the
program that is computed by tabled bottom-up evaluation in 4.1 seconds
(as shown in the figure).

\section{Conclusion}
In this paper we have described the integration of bottom-up
evaluation with top-down evaluation in tabled Prolog by means of
Subsumptive Tabling With Abstraction (STWA).  We provided a procedural
interpretation for bottom-up evaluation, in the context of Prolog
programming, as efficient indexed data structure initialization.  This
allows many formerly non-declarative Prolog programs that contain
assert/1 to be naturally expressed as declarative programs with no
sacrifice in efficiency.

We have shown how XSB provides a natural model for evaluating Horn
programs, based on nondeterministic procedural interpretation using
subsumptive tabling with abstraction, that incorporates both top-down
and bottom-up computation in a single uniform framework.  And this
model provides linear evaluation of all queries to proposition Horn
programs for all intermixings of demand and data driven computation.

%
\bibliographystyle{acmtrans}
\bibliography{tdbureconciled}

\appendix
\section{Programs to Generate Traces} \label{progs}
Program that generates trace of bottom-up evaluation of {\tt ?- p(a,X).}
\begin{verbatim}
:- table_index(p/2,[1,0]).
p(X,Y) :-
    writeln('enter rule 1: p'(X,Y)),
    e(X,Y),
    write('from fact: e'(X,Y)),
    write(' infer p'(X,Y)),
p(X,Y) :-
    writeln('enter rule 2 p'(X,Y)),
    p(X,Z),
    write('from table: p'(X,Z)),
    e(Z,Y),
    write(' and fact e'(Z,Y)),
    write(' infer p'(X,Y)).            

e(a,b).     e(b,c).
e(e,a).     e(c,b).
e(d,e).
\end{verbatim}
\normalsize

Program that generates the annotated trace of bottom-up
meta-interpretation of proposition clauses.
\begin{verbatim}
:- op(1200,xfx,('<-')).
:- import conset/2, coninc/2 from gensym.

:- table_index(interpAtom/1,[0]).
interpAtom(G) :-
    conset('_ctr',0),
    log('%S demanded',args(G)),
    (G <- Body),
    log('%S :- %S initial clause',args(G,Body)),
    interp(Body,G).

interp(true,H) :- !, log('%S from true and %S :- true',args(H,H)).
interp((A,B),H) :- !,
    interpAtom(A),
    log('%S :- %S from %S and %S :- %S',args(H,B,A,H,(A,B))),
    interp(B,H).
interp(G,H) :-
    interpAtom(G),
    log('%S from %S and %S :- %S',args(H,G,H,G)).

p <- q,v,r,s.          s <- true.
p <- q,s,t.            u <- s,p,v,r.
q <- u,r.              u <- r,t.
q <- q,t,v.            t <- true.
r <- s.
\end{verbatim}
\normalsize

\section{Subsumptive Tabling with Full Abstraction (STWFA) is Bottom-up}
\label{app-proof}
This appendix proves a close relationship between the evaluation of
Horn programs using subsumptive tabling with full abstraction and
their bottom-up evaluation.

\subsection{Multiple Machine Model of Tabled Horn Clause Evaluation}
We describe a model for STWFA evaluation as a set of (virtual)
machines.  Given a Horn clause program and a goal, each machine
carries out the evaluation of the goal in the program along a
different deterministic path.  Without tabling there would be a
machine corresponding to each root-to-leaf path of the SLD tree for
the query.  With subsumptive tabling with full abstraction, the model
is as described next.

A global table of subgoals is maintained.  Associated with each
subgoal is a set of answers corresponding to proven instances of the
subgoal.  This table is maintained and used by all the executing
machines.  It is monotonic in that goals and answers are only added to
the table and never deleted.

Computation starts with a single machine that is given the initial
goal to evaluate.  When a machine encounters a goal, it looks to see
if the corresponding most-general goal is in the global table.  If it
is {\em not} in the table, a) it adds it (with an empty list of
associated answers), b) for each clause for that most-general goal, it
forks off a duplicate of itself to process it, and c) remembering the
goal it encountered, it suspends on the table entry it just added (to
later process answers that are associated with that most-general
goal).  If the most-general goal already {\em is} in the table, the
machine (remembering its encountered goal) suspends on that table
entry.

When a machine encounters a failure, it simply disappears.  When a
machine returns an answer to a goal (i.e., completes execution to the
end of a clause for that goal), it adds the computed goal instance to
the end of the list of answers associated with that goal in the table,
and then disappears.

Whenever there is a new answer for a table entry for a goal, each machine
that is suspended on that entry, whose associated goal unifies with
that new answer, forks off a copy of itself.  This new machine uses
that answer to update its state and then continues executing.  The
suspended machine remains suspended having marked in the table that
that answer has been returned.

This set of machines is normally simulated by a sequential emulator,
which includes a scheduler that determines which of the next possible
machine operations will be performed.  At any point in the execution,
there may be a machine ready to a) evaluate the next goal of some
clause (or return an answer if there are no more goals for that
clause) or b) fork a machine that is suspended on a table entry to
process a new (unifying) answer.  For specificity and simplicity, we
will assume that answers associated with a table entry are maintained
in the order they are generated, and suspended machines return them in
that same order.

\subsection{Example of STWFA Evaluation}

To better understand this model, we trace its evaluation of the graph
reachability example, with the edge relation defined by a simple join:
\begin{verbatim}
:- table_index(p/2,[0]).
                             q(a,b).           r(a).
p(X,Y) :- e(X,Y).            q(a,d).           r(b).
p(X,Y) :- p(X,Z),e(Z,Y).     q(e,a).           r(c).
                             q(d,e). q(b,d).   r(e).
e(X,Y) :- q(X,Y),r(Y).       q(b,c). q(c,b).
\end{verbatim}
\normalsize

We trace the evaluation of the query p(a,X).  A state of a machine is
represented by a rule form: $H\leftarrow B_1,B_2,...,B_k$, where $k>=0$.
This represents a machine about to call subgoal $B_1$, which on its
return will be followed by calls to $B_2$ through $B_k$, and then a
return to $H$.  When $k=0$, it is about to return the answer H to its
call.

The table is represented by a set of forms: $Goal:[AnsList],[H\leftarrow
  B_1,B_2,...,B_k:N,...]$ where $Goal$ is a most-general goal,
$AnsList$ is a list of instances of Goal (a.k.a. answers). The indexed
rule form $H\leftarrow B_1,B_2,...,B_k:N$ represents a suspension:
$B_1$ is an instance of $Goal$, and $N$ is an integer between $0$ and
the length of the answer list, and is the index of the last answer
from $B_1$ that has been returned, i.e., for which the suspension has
been cloned and continued.  This pair of a rule and index represents a
machine that is suspended on the goal of this table entry; there may
be multiple such suspensions.

In the following trace, \verb|S:| indicates the set of machine states
that are available to be scheduled, separated by semi-colons and
\verb|T:| indicates the state of the table.  We trace the evaluation
of the query \verb|p(a,X)|.

\begin{verbatim}
S: <-p(a,X)
T:

Create new table entry for p/2, suspend on it, fork machines for each rule:
S: p(X,Y)<-e(X,Y); p(X,Y)<-p(X,Z),e(Z,Y)
T: p(X,Y):[],[<-p(a,X):0]

Create new table entry for e/2, suspend on it, fork machines for each rule:
S: e(X,Y)<-q(X,Y),r(Y); p(X,Y)<-p(X,Z),e(Z,Y)
T: p(X,Y):[],[<-p(a,X):0]
   e(X,Y):[],[p(X,Y)<-e(X,Y):0]

Create new table entry for q/2, suspend, fork machines for each rule:
S: q(a,b)<-; q(a,d)<-; q(e,a)<-; q(d,e)<-; q(b,c)<-; q(b,d)<-; q(c,b)<-;
   p(X,Y)<-p(X,Z),e(Z,Y)
T: p(X,Y):[],[<-p(a,X):0]
   e(X,Y):[],[p(X,Y)<-e(X,Y):0]
   q(X,Y):[],[e(X,Y)<-q(X,Y),r(Y):0]

7 steps, returning q/2 facts to the q(X,Y), adding them to its table:
S: p(X,Y)<-p(X,Z),e(Z,Y)
T: p(X,Y):[],[<-p(a,X):0]
   e(X,Y):[],[p(X,Y)<-e(X,Y):0]
   q(X,Y):[q(a,b),q(a,d),q(e,a),q(d,e),q(b,c),q(b,d),q(c,b)],
          [e(X,Y)<-q(X,Y),r(Y):0]

Fork off 7 copies of the suspended e<-q,r machine, one for each answer:
S: e(a,b)<-r(b); e(a,d)<-r(d); e(e,a)<-r(a); e(d,e)<-r(e); e(b,c)<-r(c);
   e(b,d)<-r(d); e(c,b)<-r(b); p(X,Y)<-p(X,Z),e(Z,Y)
T: p(X,Y):[],[<-p(a,X):0]
   e(X,Y):[],[p(X,Y)<-e(X,Y):0]
   q(X,Y):[q(a,b),q(a,d),q(e,a),q(d,e),q(b,c),q(b,d),q(c,b)],
          [e(X,Y)<-q(X,Y),r(Y):7]

Create new table entry for r/1, suspend on it, fork one for each r/1 rule:
S: r(a)<-; r(b)<-; r(c)<-; r(e)<-;
   e(a,d)<-r(d); e(e,a)<-r(a); e(d,e)<-r(e); e(b,c)<-r(c); e(b,d)<-r(d)
   e(c,b)<-r(b); p(X,Y)<-p(X,Z),e(Z,Y)
T: p(X,Y):[],[<-p(a,X):0]
   e(X,Y):[],[p(X,Y)<-e(X,Y):0]a
   q(X,Y):[q(a,b),q(a,d),q(e,a),q(d,e),q(b,c),q(b,d),q(c,b)],
          [e(X,Y)<-q(X,Y),r(Y):7]
   r(X):[],[e(a,b)<-r(b):0]

4 steps, returning 4 answers to r/1 call:
S: e(a,d)<-r(d); e(e,a)<-r(a); e(d,e)<-r(e); e(b,c)<-r(c); e(b,d)<-r(d);
   e(c,b)<-r(b); p(X,Y)<-p(X,Z),e(Z,Y)
T: p(X,Y):[],[<-p(a,X):0]
   e(X,Y):[],[p(X,Y)<-e(X,Y):0]
   q(X,Y):[q(a,b),q(a,d),q(e,a),q(d,e),q(b,c),q(b,d),q(c,b)],
          [e(X,Y)<-q(X,Y),r(Y):7]
   r(X):[r(a),r(b),r(c),r(e)],[e(a,b)<-r(b):0]

7 steps each calling r/1 and suspending on that call:
S: p(X,Y)<-p(X,Z),e(Z,Y)
T: p(X,Y):[],[<-p(a,X):0]
   e(X,Y):[],[p(X,Y)<-e(X,Y):0]
   q(X,Y):[q(a,b),q(a,d),q(e,a),q(d,e),q(b,c),q(b,d),q(c,b)],
          [e(X,Y)<-q(X,Y),r(Y):7]
   r(X):[r(a),r(b),r(c),r(e)],
        [e(a,b)<-r(b):0, e(a,d)<-r(d):0, e(e,a)<-r(a):0,
         e(d,e)<-r(e):0, e(b,c)<-r(c):0, e(b,d)<-r(d):0,
         e(c,b)<-r(b):0]

7 (*4) steps returning r/1 answers to suspended machine.  5 of the
  suspended machines match on answer and 2 match no answers, generating
  5 new machines:
S: e(a,b)<-; e(e,a)<-; e(d,e)<-; e(b,c)<-; e(c,b)<-; 
p(X,Y)<-p(X,Z),e(Z,Y)
T: p(X,Y):[],[<-p(a,X):0]
   e(X,Y):[],[p(X,Y)<-e(X,Y):0]
   q(X,Y):[q(a,b),q(a,d),q(e,a),q(d,e),q(b,c),q(b,d),q(c,b)],
          [e(X,Y)<-q(X,Y),r(Y):7]
   r(X):[r(a),r(b),r(c),r(e)],
        [e(a,b)<-r(b):4, e(a,d)<-r(d):4, e(e,a)<-r(a):4,
         e(d,e)<-r(e):4, e(b,c)<-r(c):4, e(b,d)<-r(d):4,
         e(c,b)<-r(b):4]

5 steps returning 5 answers to the e/2 table:
S: p(X,Y)<-p(X,Z),e(Z,Y)
T: p(X,Y):[],[<-p(a,X):0]
   e(X,Y):[e(a,b),e(e,a),e(d,e),e(b,c),e(c,b)],[p(X,Y)<-e(X,Y):0]
   q(X,Y):[q(a,b),q(a,d),q(e,a),q(d,e),q(b,c),q(b,d),q(c,b)],
          [e(X,Y)<-q(X,Y),r(Y):7]
   r(X):[r(a),r(b),r(c),r(e)],
        [e(a,b)<-r(b):4, e(a,d)<-r(d):4, e(e,a)<-r(a):4,
         e(d,e)<-r(e):4, e(b,c)<-r(c):4, e(b,d)<-r(d):4,
         e(c,b)<-r(b):4]

The machine suspended on e/2 forks 5 new machines, one for each new answer:
S: p(a,b)<-; p(e,a)<- ; p(d,e)<-; p(b,c)<-; p(c,b)<-; p(X,Y)<-p(X,Z),e(Z,Y)
T: p(X,Y):[],[<-p(a,X):0]
   e(X,Y):[e(a,b),e(e,a),e(d,e),e(b,c),e(c,b)],[p(X,Y)<-e(X,Y):5]
   q(X,Y):[q(a,b),q(a,d),q(e,a),q(d,e),q(b,c),q(b,d),q(c,b)],
          [e(X,Y)<-q(X,Y),r(Y):7]
   r(X):[r(a),r(b),r(c),r(e)],
        [e(a,b)<-r(b):4, e(a,d)<-r(d):4, e(e,a)<-r(a):4,
         e(d,e)<-r(e):4, e(b,c)<-r(c):4, e(b,d)<-r(d):4,
         e(c,b)<-r(b):4]

5 steps, adding 5 new answers to p/2 entry:
S: p(X,Y)<-p(X,Z),e(Z,Y)
T: p(X,Y):[p(a,b),p(e,a),p(d,e),p(b,c),p(c,b)],[<-p(a,X):0]
   e(X,Y):[e(a,b),e(e,a),e(d,e),e(b,c),e(c,b)],[p(X,Y)<-e(X,Y):5]
   q(X,Y):[q(a,b),q(a,d),q(e,a),q(d,e),q(b,c),q(b,d),q(c,b)],
          [e(X,Y)<-q(X,Y),r(Y):7]
   r(X):[r(a),r(b),r(c),r(e)],
        [e(a,b)<-r(b):4, e(a,d)<-r(d):4, e(e,a)<-r(a):4,
         e(d,e)<-r(e):4, e(b,c)<-r(c):4, e(b,d)<-r(d):4,
         e(c,b)<-r(b):4]

Machine suspends on call to p(X,Z):
S: 
T: p(X,Y):[p(a,b),p(e,a),p(d,e),p(b,c),p(c,b)],[<-p(a,X):0,p(X,Y)<-p(X,Z),e(Z,Y):0]
   e(X,Y):[e(a,b),e(e,a),e(d,e),e(b,c),e(c,b)],[p(X,Y)<-e(X,Y):5]
   q(X,Y):[q(a,b),q(a,d),q(e,a),q(d,e),q(b,c),q(b,d),q(c,b)],
          [e(X,Y)<-q(X,Y),r(Y):7]
   r(X):[r(a),r(b),r(c),r(e)],
        [e(a,b)<-r(b):4, e(a,d)<-r(d):4, e(e,a)<-r(a):4,
         e(d,e)<-r(e):4, e(b,c)<-r(c):4, e(b,d)<-r(d):4,
         e(c,b)<-r(b):4]

5 new machines forked from machine suspended on p/2:
S: p(a,Y)<-e(b,Y); p(e,Y)<-e(a,Y); p(d,Y)<-e(e,Y); p(b,Y)<-e(c,Y); p(c,Y)<-e(b,Y)        
T: p(X,Y):[p(a,b),p(e,a),p(d,e),p(b,c),p(c,b)],[<-p(a,X):0,p(X,Y)<-p(X,Z),e(Z,Y):5]
   e(X,Y):[e(a,b),e(e,a),e(d,e),e(b,c),e(c,b)],[p(X,Y)<-e(X,Y):5]
   q(X,Y):[q(a,b),q(a,d),q(e,a),q(d,e),q(b,c),q(b,d),q(c,b)],
          [e(X,Y)<-q(X,Y),r(Y):7]
   r(X):[r(a),r(b),r(c),r(e)],
        [e(a,b)<-r(b):4, e(a,d)<-r(d):4, e(e,a)<-r(a):4,
         e(d,e)<-r(e):4, e(b,c)<-r(c):4, e(b,d)<-r(d):4,
         e(c,b)<-r(b):4]

All 5 just generated machines suspend on their calls to e/2:
S: 
T: p(X,Y):[p(a,b),p(e,a),p(d,e),p(b,c),p(c,b)],[<-p(a,X):0,p(X,Y)<-p(X,Z),e(Z,Y):5]
   e(X,Y):[e(a,b),e(e,a),e(d,e),e(b,c),e(c,b)],
          [p(X,Y)<-e(X,Y):5, p(a,Y)<-e(b,Y):0, p(e,Y)<-e(a,Y):0,
           p(d,Y)<-e(e,Y):0, p(b,Y)<-e(c,Y):0, p(c,Y)<-e(b,Y):0]
   q(X,Y):[q(a,b),q(a,d),q(e,a),q(d,e),q(b,c),q(b,d),q(c,b)],
          [e(X,Y)<-q(X,Y),r(Y):7]
   r(X):[r(a),r(b),r(c),r(e)],
        [e(a,b)<-r(b):4, e(a,d)<-r(d):4, e(e,a)<-r(a):4,
         e(d,e)<-r(e):4, e(b,c)<-r(c):4, e(b,d)<-r(d):4,
         e(c,b)<-r(b):4]

The 5 just suspended machines fork of 5 new machines; each machine
  matching exactly one of the 5 new answers:
S: p(a,c)<-; p(e,b)<-; p(d,a)<-; p(b,b)<-; p(c,c)<-
T: p(X,Y):[p(a,b),p(e,a),p(d,e),p(b,c),p(c,b)],[<-p(a,X):0,p(X,Y)<-p(X,Z),e(Z,Y):5]
   e(X,Y):[e(a,b),e(e,a),e(d,e),e(b,c),e(c,b)],
          [p(X,Y)<-e(X,Y):5, p(a,Y)<-e(b,Y):5, p(e,Y)<-e(a,Y):5,
           p(d,Y)<-e(e,Y):5, p(b,Y)<-e(c,Y):5, p(c,Y)<-e(b,Y):5]
   q(X,Y):[q(a,b),q(a,d),q(e,a),q(d,e),q(b,c),q(b,d),q(c,b)],
          [e(X,Y)<-q(X,Y),r(Y):7]
   r(X):[r(a),r(b),r(c),r(e)],
        [e(a,b)<-r(b):4, e(a,d)<-r(d):4, e(e,a)<-r(a):4,
         e(d,e)<-r(e):4, e(b,c)<-r(c):4, e(b,d)<-r(d):4,
         e(c,b)<-r(b):4]

Each of the 5 new machines add a new answer to the p/2 table entry:
S:
T: p(X,Y):[p(a,b),p(e,a),p(d,e),p(b,c),p(c,b),p(a,c),p(e,b),p(d,a),p(b,b),p(c,c)],
             [<-p(a,X):0,p(X,Y)<-p(X,Z),e(Z,Y):5]
   e(X,Y):[e(a,b),e(e,a),e(d,e),e(b,c),e(c,b)],
          [p(X,Y)<-e(X,Y):5, p(a,Y)<-e(b,Y):5, p(e,Y)<-e(a,Y):5,
           p(d,Y)<-e(e,Y):5, p(b,Y)<-e(c,Y):5, p(c,Y)<-e(b,Y):5]
   q(X,Y):[q(a,b),q(a,d),q(e,a),q(d,e),q(b,c),q(b,d),q(c,b)],
          [e(X,Y)<-q(X,Y),r(Y):7]
   r(X):[r(a),r(b),r(c),r(e)],
        [e(a,b)<-r(b):4, e(a,d)<-r(d):4, e(e,a)<-r(a):4,
         e(d,e)<-r(e):4, e(b,c)<-r(c):4, e(b,d)<-r(d):4,
         e(c,b)<-r(b):4]

The machine suspended on p(X,Z) in the p/2 table forks off machines
  for each of the 5 just-add p/2 answers:
S: p(a,Y)<-e(c,Y); p(e,Y)<-e(b,Y); p(d,Y)<-e(a,Y); p(b,Y)<-e(b,Y); p(c,Y)<-e(c,Y)
T: p(X,Y):[p(a,b),p(e,a),p(d,e),p(b,c),p(c,b),p(a,c),p(e,b),p(d,a),p(b,b),p(c,c)],
             [<-p(a,X):0,p(X,Y)<-p(X,Z),e(Z,Y):10]
   e(X,Y):[e(a,b),e(e,a),e(d,e),e(b,c),e(c,b)],
          [p(X,Y)<-e(X,Y):5, p(a,Y)<-e(b,Y):5, p(e,Y)<-e(a,Y):5,
           p(d,Y)<-e(e,Y):5, p(b,Y)<-e(c,Y):5, p(c,Y)<-e(b,Y):5]
   q(X,Y):[q(a,b),q(a,d),q(e,a),q(d,e),q(b,c),q(b,d),q(c,b)],
          [e(X,Y)<-q(X,Y),r(Y):7]
   r(X):[r(a),r(b),r(c),r(e)],
        [e(a,b)<-r(b):4, e(a,d)<-r(d):4, e(e,a)<-r(a):4,
         e(d,e)<-r(e):4, e(b,c)<-r(c):4, e(b,d)<-r(d):4,
         e(c,b)<-r(b):4]

All 5 of those machines suspend on their calls to e/2.
S: 
T: p(X,Y):[p(a,b),p(e,a),p(d,e),p(b,c),p(c,b),p(a,c),p(e,b),p(d,a),p(b,b),p(c,c)],
             [<-p(a,X):0,p(X,Y)<-p(X,Z),e(Z,Y):10]
   e(X,Y):[e(a,b),e(e,a),e(d,e),e(b,c),e(c,b)],
          [p(X,Y)<-e(X,Y):5, p(a,Y)<-e(b,Y):5, p(e,Y)<-e(a,Y):5,
           p(d,Y)<-e(e,Y):5, p(b,Y)<-e(c,Y):5, p(c,Y)<-e(b,Y):5,
           p(a,Y)<-e(c,Y):0, p(e,Y)<-e(b,Y):0, p(d,Y)<-e(a,Y):0,
           p(b,Y)<-e(b,Y):0,  p(c,Y)<-e(c,Y):0]
   q(X,Y):[q(a,b),q(a,d),q(e,a),q(d,e),q(b,c),q(b,d),q(c,b)],
          [e(X,Y)<-q(X,Y),r(Y):7]
   r(X):[r(a),r(b),r(c),r(e)],
        [e(a,b)<-r(b):4, e(a,d)<-r(d):4, e(e,a)<-r(a):4,
         e(d,e)<-r(e):4, e(b,c)<-r(c):4, e(b,d)<-r(d):4,
         e(c,b)<-r(b):4]

All 5 of those suspensions fork off a total of 5 new machines, each
  one matching one of the 5 existing answers for e/2:
S: p(a,b)<-; p(e,c)<-; p(d,b)<-; p(b,c)<-; p(c,b)<-       
T: p(X,Y):[p(a,b),p(e,a),p(d,e),p(b,c),p(c,b),p(a,c),p(e,b),p(d,a),p(b,b),p(c,c)],
             [<-p(a,X):0,p(X,Y)<-p(X,Z),e(Z,Y):10]
   e(X,Y):[e(a,b),e(e,a),e(d,e),e(b,c),e(c,b)],
          [p(X,Y)<-e(X,Y):5, p(a,Y)<-e(b,Y):5, p(e,Y)<-e(a,Y):5,
           p(d,Y)<-e(e,Y):5, p(b,Y)<-e(c,Y):5, p(c,Y)<-e(b,Y):5,
           p(a,Y)<-e(c,Y):5, p(e,Y)<-e(b,Y):5, p(d,Y)<-e(a,Y):5,
           p(b,Y)<-e(b,Y):5, p(c,Y)<-e(c,Y):5]
   q(X,Y):[q(a,b),q(a,d),q(e,a),q(d,e),q(b,c),q(b,d),q(c,b)],
          [e(X,Y)<-q(X,Y),r(Y):7]
   r(X):[r(a),r(b),r(c),r(e)],
        [e(a,b)<-r(b):4, e(a,d)<-r(d):4, e(e,a)<-r(a):4,
         e(d,e)<-r(e):4, e(b,c)<-r(c):4, e(b,d)<-r(d):4,
         e(c,b)<-r(b):4]

Those 5 machines return just 2 new answers to the p/2 entry (3 are duplicates):
S:
T: p(X,Y):[p(a,b),p(e,a),p(d,e),p(b,c),p(c,b),p(a,c),p(e,b),p(d,a),p(b,b),p(c,c),
              p(e,c),p(d,b)],
             [<-p(a,X):0,p(X,Y)<-p(X,Z),e(Z,Y):10]
   e(X,Y):[e(a,b),e(e,a),e(d,e),e(b,c),e(c,b)],
          [p(X,Y)<-e(X,Y):5, p(a,Y)<-e(b,Y):5, p(e,Y)<-e(a,Y):5,
           p(d,Y)<-e(e,Y):5, p(b,Y)<-e(c,Y):5, p(c,Y)<-e(b,Y):5,
           p(a,Y)<-e(c,Y):5, p(e,Y)<-e(b,Y):5, p(d,Y)<-e(a,Y):5,
           p(b,Y)<-e(b,Y):5, p(c,Y)<-e(c,Y):5]
   q(X,Y):[q(a,b),q(a,d),q(e,a),q(d,e),q(b,c),q(b,d),q(c,b)],
          [e(X,Y)<-q(X,Y),r(Y):7]
   r(X):[r(a),r(b),r(c),r(e)],
        [e(a,b)<-r(b):4, e(a,d)<-r(d):4, e(e,a)<-r(a):4,
         e(d,e)<-r(e):4, e(b,c)<-r(c):4, e(b,d)<-r(d):4,
         e(c,b)<-r(b):4]

The suspend p<-p,e machine forks 2 new machines, one for each just-added answer;
S: p(e,Y)<-e(c,Y); p(d,Y)<-e(b,Y)
T: p(X,Y):[p(a,b),p(e,a),p(d,e),p(b,c),p(c,b),p(a,c),p(e,b),p(d,a),p(b,b),p(c,c),
              p(e,c),p(d,b)],
             [<-p(a,X):0,p(X,Y)<-p(X,Z),e(Z,Y):12]
   e(X,Y):[e(a,b),e(e,a),e(d,e),e(b,c),e(c,b)],
          [p(X,Y)<-e(X,Y):5, p(a,Y)<-e(b,Y):5, p(e,Y)<-e(a,Y):5,
           p(d,Y)<-e(e,Y):5, p(b,Y)<-e(c,Y):5, p(c,Y)<-e(b,Y):5,
           p(a,Y)<-e(c,Y):5, p(e,Y)<-e(b,Y):5, p(d,Y)<-e(a,Y):5,
           p(b,Y)<-e(b,Y):5, p(c,Y)<-e(c,Y):5]
   q(X,Y):[q(a,b),q(a,d),q(e,a),q(d,e),q(b,c),q(b,d),q(c,b)],
          [e(X,Y)<-q(X,Y),r(Y):7]
   r(X):[r(a),r(b),r(c),r(e)],
        [e(a,b)<-r(b):4, e(a,d)<-r(d):4, e(e,a)<-r(a):4,
         e(d,e)<-r(e):4, e(b,c)<-r(c):4, e(b,d)<-r(d):4,
         e(c,b)<-r(b):4]

Those 2 new machines suspend on their e/2 calls:
S: 
T: p(X,Y):[p(a,b),p(e,a),p(d,e),p(b,c),p(c,b),p(a,c),p(e,b),p(d,a),p(b,b),p(c,c),
              p(e,c),p(d,b)],
             [<-p(a,X):0,p(X,Y)<-p(X,Z),e(Z,Y):12]
   e(X,Y):[e(a,b),e(e,a),e(d,e),e(b,c),e(c,b)],
          [p(X,Y)<-e(X,Y):5, p(a,Y)<-e(b,Y):5, p(e,Y)<-e(a,Y):5,
           p(d,Y)<-e(e,Y):5, p(b,Y)<-e(c,Y):5, p(c,Y)<-e(b,Y):5,
           p(a,Y)<-e(c,Y):5, p(e,Y)<-e(b,Y):5, p(d,Y)<-e(a,Y):5,
           p(b,Y)<-e(b,Y):5, p(c,Y)<-e(c,Y):5, p(e,Y)<-e(c,Y):0
           p(d,Y)<-e(b,Y):0]
   q(X,Y):[q(a,b),q(a,d),q(e,a),q(d,e),q(b,c),q(b,d),q(c,b)],
          [e(X,Y)<-q(X,Y),r(Y):7]
   r(X):[r(a),r(b),r(c),r(e)],
        [e(a,b)<-r(b):4, e(a,d)<-r(d):4, e(e,a)<-r(a):4,
         e(d,e)<-r(e):4, e(b,c)<-r(c):4, e(b,d)<-r(d):4,
         e(c,b)<-r(b):4]

Those 2 suspensions fork 2 new machines from all the e/2 answers,
  each matching exactly one:
S: p(e,b)<-; p(d,c)<-
T: p(X,Y):[p(a,b),p(e,a),p(d,e),p(b,c),p(c,b),p(a,c),p(e,b),p(d,a),p(b,b),p(c,c),
              p(e,c),p(d,b)],
             [<-p(a,X):0,p(X,Y)<-p(X,Z),e(Z,Y):12]
   e(X,Y):[e(a,b),e(e,a),e(d,e),e(b,c),e(c,b)],
          [p(X,Y)<-e(X,Y):5, p(a,Y)<-e(b,Y):5, p(e,Y)<-e(a,Y):5,
           p(d,Y)<-e(e,Y):5, p(b,Y)<-e(c,Y):5, p(c,Y)<-e(b,Y):5,
           p(a,Y)<-e(c,Y):5, p(e,Y)<-e(b,Y):5, p(d,Y)<-e(a,Y):5,
           p(b,Y)<-e(b,Y):5, p(c,Y)<-e(c,Y):5, p(e,Y)<-e(c,Y):5,
           p(d,Y)<-e(b,Y):5]
   q(X,Y):[q(a,b),q(a,d),q(e,a),q(d,e),q(b,c),q(b,d),q(c,b)],
          [e(X,Y)<-q(X,Y),r(Y):7]
   r(X):[r(a),r(b),r(c),r(e)],
        [e(a,b)<-r(b):4, e(a,d)<-r(d):4, e(e,a)<-r(a):4,
         e(d,e)<-r(e):4, e(b,c)<-r(c):4, e(b,d)<-r(d):4,
         e(c,b)<-r(b):4]

These two machines add just one new answer to the p/2 table entry:
S:
T: p(X,Y):[p(a,b),p(e,a),p(d,e),p(b,c),p(c,b),p(a,c),p(e,b),p(d,a),p(b,b),p(c,c),
              p(e,c),p(d,b),p(d,c)],
             [<-p(a,X):0,p(X,Y)<-p(X,Z),e(Z,Y):12]
   e(X,Y):[e(a,b),e(e,a),e(d,e),e(b,c),e(c,b)],
          [p(X,Y)<-e(X,Y):5, p(a,Y)<-e(b,Y):5, p(e,Y)<-e(a,Y):5,
           p(d,Y)<-e(e,Y):5, p(b,Y)<-e(c,Y):5, p(c,Y)<-e(b,Y):5,
           p(a,Y)<-e(c,Y):5, p(e,Y)<-e(b,Y):5, p(d,Y)<-e(a,Y):5,
           p(b,Y)<-e(b,Y):5, p(c,Y)<-e(c,Y):5, p(e,Y)<-e(c,Y):5,
           p(d,Y)<-e(b,Y):5]
   q(X,Y):[q(a,b),q(a,d),q(e,a),q(d,e),q(b,c),q(b,d),q(c,b)],
          [e(X,Y)<-q(X,Y),r(Y):7]
   r(X):[r(a),r(b),r(c),r(e)],
        [e(a,b)<-r(b):4, e(a,d)<-r(d):4, e(e,a)<-r(a):4,
         e(d,e)<-r(e):4, e(b,c)<-r(c):4, e(b,d)<-r(d):4,
         e(c,b)<-r(b):4]

The p<-p,e suspension on p/2 forks off a new machine for that new answer:
S: p(d,Y)<-e(c,Y)
T: p(X,Y):[p(a,b),p(e,a),p(d,e),p(b,c),p(c,b),p(a,c),p(e,b),p(d,a),p(b,b),p(c,c),
              p(e,c),p(d,b),p(d,c)],
             [<-p(a,X):0,p(X,Y)<-p(X,Z),e(Z,Y):13]
   e(X,Y):[e(a,b),e(e,a),e(d,e),e(b,c),e(c,b)],
          [p(X,Y)<-e(X,Y):5, p(a,Y)<-e(b,Y):5, p(e,Y)<-e(a,Y):5,
           p(d,Y)<-e(e,Y):5, p(b,Y)<-e(c,Y):5, p(c,Y)<-e(b,Y):5,
           p(a,Y)<-e(c,Y):5, p(e,Y)<-e(b,Y):5, p(d,Y)<-e(a,Y):5,
           p(b,Y)<-e(b,Y):5, p(c,Y)<-e(c,Y):5, p(e,Y)<-e(c,Y):5,
           p(d,Y)<-e(b,Y):5]
   q(X,Y):[q(a,b),q(a,d),q(e,a),q(d,e),q(b,c),q(b,d),q(c,b)],
          [e(X,Y)<-q(X,Y),r(Y):7]
   r(X):[r(a),r(b),r(c),r(e)],
        [e(a,b)<-r(b):4, e(a,d)<-r(d):4, e(e,a)<-r(a):4,
         e(d,e)<-r(e):4, e(b,c)<-r(c):4, e(b,d)<-r(d):4,
         e(c,b)<-r(b):4]

That machine calls e/2 and suspends:
S: 
T: p(X,Y):[p(a,b),p(e,a),p(d,e),p(b,c),p(c,b),p(a,c),p(e,b),p(d,a),p(b,b),p(c,c),
              p(e,c),p(d,b),p(d,c)],
             [<-p(a,X):0,p(X,Y)<-p(X,Z),e(Z,Y):13]
   e(X,Y):[e(a,b),e(e,a),e(d,e),e(b,c),e(c,b)],
          [p(X,Y)<-e(X,Y):5, p(a,Y)<-e(b,Y):5, p(e,Y)<-e(a,Y):5,
           p(d,Y)<-e(e,Y):5, p(b,Y)<-e(c,Y):5, p(c,Y)<-e(b,Y):5,
           p(a,Y)<-e(c,Y):5, p(e,Y)<-e(b,Y):5, p(d,Y)<-e(a,Y):5,
           p(b,Y)<-e(b,Y):5, p(c,Y)<-e(c,Y):5, p(e,Y)<-e(c,Y):5,
           p(d,Y)<-e(b,Y):5, p(d,Y)<-e(c,Y):0]
   q(X,Y):[q(a,b),q(a,d),q(e,a),q(d,e),q(b,c),q(b,d),q(c,b)],
          [e(X,Y)<-q(X,Y),r(Y):7]
   r(X):[r(a),r(b),r(c),r(e)],
        [e(a,b)<-r(b):4, e(a,d)<-r(d):4, e(e,a)<-r(a):4,
         e(d,e)<-r(e):4, e(b,c)<-r(c):4, e(b,d)<-r(d):4,
         e(c,b)<-r(b):4]

That suspension forks one machine for the one matching answer in e/2:
S: p(d,b)<-
T: p(X,Y):[p(a,b),p(e,a),p(d,e),p(b,c),p(c,b),p(a,c),p(e,b),p(d,a),p(b,b),p(c,c),
              p(e,c),p(d,b),p(d,c)],
             [<-p(a,X):0,p(X,Y)<-p(X,Z),e(Z,Y):13]
   e(X,Y):[e(a,b),e(e,a),e(d,e),e(b,c),e(c,b)],
          [p(X,Y)<-e(X,Y):5, p(a,Y)<-e(b,Y):5, p(e,Y)<-e(a,Y):5,
           p(d,Y)<-e(e,Y):5, p(b,Y)<-e(c,Y):5, p(c,Y)<-e(b,Y):5,
           p(a,Y)<-e(c,Y):5, p(e,Y)<-e(b,Y):5, p(d,Y)<-e(a,Y):5,
           p(b,Y)<-e(b,Y):5, p(c,Y)<-e(c,Y):5, p(e,Y)<-e(c,Y):5,
           p(d,Y)<-e(b,Y):5, p(d,Y)<-e(c,Y):5]
   q(X,Y):[q(a,b),q(a,d),q(e,a),q(d,e),q(b,c),q(b,d),q(c,b)],
          [e(X,Y)<-q(X,Y),r(Y):7]
   r(X):[r(a),r(b),r(c),r(e)],
        [e(a,b)<-r(b):4, e(a,d)<-r(d):4, e(e,a)<-r(a):4,
         e(d,e)<-r(e):4, e(b,c)<-r(c):4, e(b,d)<-r(d):4,
         e(c,b)<-r(b):4]

That machine fails since the answer it wants to add is a duplicate:
S: 
T: p(X,Y):[p(a,b),p(e,a),p(d,e),p(b,c),p(c,b),p(a,c),p(e,b),p(d,a),p(b,b),p(c,c),
              p(e,c),p(d,b),p(d,c)],
             [<-p(a,X):0,p(X,Y)<-p(X,Z),e(Z,Y):13]
   e(X,Y):[e(a,b),e(e,a),e(d,e),e(b,c),e(c,b)],
          [p(X,Y)<-e(X,Y):5, p(a,Y)<-e(b,Y):5, p(e,Y)<-e(a,Y):5,
           p(d,Y)<-e(e,Y):5, p(b,Y)<-e(c,Y):5, p(c,Y)<-e(b,Y):5,
           p(a,Y)<-e(c,Y):5, p(e,Y)<-e(b,Y):5, p(d,Y)<-e(a,Y):5,
           p(b,Y)<-e(b,Y):5, p(c,Y)<-e(c,Y):5, p(e,Y)<-e(c,Y):5,
           p(d,Y)<-e(b,Y):5, p(d,Y)<-e(c,Y):5]
   q(X,Y):[q(a,b),q(a,d),q(e,a),q(d,e),q(b,c),q(b,d),q(c,b)],
          [e(X,Y)<-q(X,Y),r(Y):7]
   r(X):[r(a),r(b),r(c),r(e)],
        [e(a,b)<-r(b):4, e(a,d)<-r(d):4, e(e,a)<-r(a):4,
         e(d,e)<-r(e):4, e(b,c)<-r(c):4, e(b,d)<-r(d):4,
         e(c,b)<-r(b):4]
\end{verbatim}
\normalsize

At this point, the only thing that can be scheduled is the
\verb|<-p(a,X)| suspension, which would match the two answers in
the table with first field {\tt a} and terminate.

\subsection{Relationship of STWFA and Bottom-Up}

It is clear that, in general, STWFA is not exactly bottom-up least
model computation, since STWFA depends on the initial goal.  For
example, consider the program:
\begin{verbatim}
  p(X) :- q(X).              q(a).
  r(X) :- q(X),p(X).         q(b).
\end{verbatim}
\normalsize
and the query \verb|p(X)|. In this case, \verb|r(X)| will never be
called and so with STWFA its extension will not be evaluated, whereas
bottom-up least-model evaluation would generate it.

Consider also the example:
\begin{verbatim}
  p(X) :- q(X),r(X),s(X).       q(a).
                                r(b).
                                s(c).
\end{verbatim}
\normalsize
and query {\tt p(X)}.  In this case, even though {\tt s/1} is
reachable from {\tt p/1} in the dependency graph of this program, {\tt
  s(X)} will never be called under STWFA since there is no {\tt X}
such that {\tt q(X)} and {\tt r(X)} is true.  So here again is a case
where bottom-up, which would, of course, compute the extension of
\verb|s/1|, differs from STWFA.

If these possibilities are excluded, we can show a close relationship
between STWFA and bottom-up evaluation.\\

\noindent {\it Theorem 1}\\
Assume H is a Horn program and P a predicate in H such that all
predicates in H are reachable from P in the predicate call graph of
H. (I.e., the goal predicate potentially depends on all predicates.)
Assume that for every rule body in H, there is some instance of it
that is true in the least model of H.  Then STWFA evaluation of
predicate P in H is equivalent to a bottom-up computation of the least
model of H.

\begin{proof}
(Sketch) We formulate a bottom-up Horn clause evaluator and then
  compare it to the multiple machine model of STWFA described above.

Consider the following formulation of a bottom-up Horn clause
evaluator.  For each clause, there is a process to generate facts for
instances of its head atom.  There is a global set of proven facts
that is initialized to empty.  The rule processes run as follows: Each
rule process looks to find an instance of its clause such that all its
body atoms appear in the global fact set.  For each such clause
instance, it adds the head atom to the global fact set, if it is not
already a member.  The processes continue to run to a fixpoint, i.e.,
until no process can add any new fact to the global set.

By synchronizing and scheduling the rule processes so that they
iteratively generate a new set of facts to be added to the fact set by
using facts only from earlier iterations, one obtains the well-known
iterative bottom-up evaluation algorithm.

However, any order of execution of the rule processes is possible, and
constitutes an evaluation method appropriately called bottom-up.

In STWFA, a machine that is generated for a rule when its most-general
subgoal is first entered corresponds to a rule process in the
bottom-up algorithm.  The execution of such a machine (and its
descendants), which suspend and fork off new machines to process new
answers, corresponds to the bottom-up rule process for the clause
that finds rule instances whose body atoms are in the global fact set.
Finally, a machine returning a new answer to a predicate corresponds
to a rule process adding a new fact to the global fact set.

The difference between the two algorithms is how the rule
processes/machines get set up in the first place.  For the bottom-up
algorithm, the rule processes are set up at the same time during an
initialization phase.  With STWFA, the corresponding machines for the
rules of a particular predicate are generated at the first call to a
goal for that predicate.  STWFA will generate relations only for
predicates on which the goal predicate depends, since it will not
initialize rule processor machines for other predicates that are never
called.

The theorem requires that the goal depends on all predicates in the
program, and that every clause has a true instance in the least model
of the program.  This will guarantee that STWFA will eventually call
every predicate and so all rule-processing machines will be
initialized.

We can fix a scheduling order for STWFA, and then find an equivalent
scheduling order for bottom-up evaluation with rule processes. 
\end{proof}

Of course, if the conditions of the theorem are not met, STWFA still
computes the correct answers to the query, and computes them in a
bottom-up manner; it just does not compute relations that are
unnecessary (given its left-to-right selection strategy).

\end{document}